\documentclass[11pt]{article}
\usepackage[margin=1in]{geometry}
\usepackage{amsthm,amsfonts,amsmath,amssymb,mathtools}
\usepackage{enumerate,graphicx}
\usepackage{xspace}
\usepackage{boxedminipage}
\usepackage{url}
\usepackage{accents}
\usepackage{multirow}
\usepackage{booktabs}
\usepackage{enumitem}
\usepackage{fullpage}
\usepackage{tikz}
\usepackage{natbib}
\usepackage[bookmarks=false,
breaklinks=false,pdfborder={0 0 1},colorlinks=true,allcolors=blue]{hyperref} %
\usepackage{placeins}

\graphicspath{{img/}}

\usepackage{style}

\usepackage{algorithm}
\usepackage{algpseudocode}

\author{
    Chirag Pabbaraju\thanks{Stanford University. Email: \texttt{cpabbara@cs.stanford.edu.}}
    \and
    Ali Vakilian\thanks{Toyota Technological Institute at Chicago. Email: \texttt{vakilian@ttic.edu.}}
}

\title{New and Improved Bounds for Markov Paging}
\date{\today}

\begin{document}

\maketitle

\begin{abstract}
    In the Markov paging model, one assumes that page requests are drawn from a Markov chain over the pages in memory, and the goal is to maintain a fast cache that suffers few page faults in expectation. While computing the optimal online algorithm $(\mathrm{OPT})$ for this problem naively takes time exponential in the size of the cache, the best-known polynomial-time approximation algorithm is the dominating distribution algorithm due to Lund, Phillips and Reingold (FOCS 1994), who showed that the algorithm is $4$-competitive against $\mathrm{OPT}$. We substantially improve their analysis and show that the dominating distribution algorithm is in fact $2$-competitive against $\mathrm{OPT}$. We also show a lower bound of $1.5907$-competitiveness for this algorithm---to the best of our knowledge, no such lower bound was previously known.
\end{abstract}

\newpage
\section{Introduction}
\label{sec:intro}

The online paging problem is a fundamental problem in %
the area of online
algorithms. There are $n$ pages in slow memory, and requests for these pages arrive in a sequential manner. We are allowed to maintain a fast cache of size $k$, which initially comprises of some $k$ pages. At time step $t$, if the page that is requested, say $p_t$, exists in the cache, this corresponds to a cache hit: we suffer a zero cost, the cache stays as is, and we move on to the next request. Otherwise, we incur a cache miss/page fault, suffer a unit cost, and are forced to evict some page from the cache so as to bring $p_t$ from the slow memory into the cache. A paging algorithm/policy is then specified by how it chooses the page to evict whenever it suffers a cache miss.

An optimal offline algorithm for this problem is the algorithm that has entire knowledge of the page request sequence, and thereafter makes its eviction choices in a way that minimizes the total number of cache misses it incurs. The classical Farthest-in-Future algorithm \citep{belady1966study}, which at any cache miss, evicts the page that is next requested latest in the remainder of the sequence, is known to be an optimal offline policy. In the spirit of competitive analysis as introduced by \cite{sleator1985amortized}, one way to benchmark any online algorithm (i.e., one which does not have knowledge of the page request sequence) is to compute the ratio of the number of page faults that the algorithm suffers, to the number of page faults that the optimal offline algorithm suffers, on a \textit{worst-case} sequence of page requests. It is known that any deterministic online algorithm has a worst-case ratio that is at least $k$, while any randomized online algorithm has a worst-case ratio of $\Omega(\log k)$.

Taking a slightly more optimistic view, one can also consider benchmarking an algorithm against the best \textit{online} algorithm, that does not know beforehand the realized sequence of page requests. Additionally, in the context of beyond worst-case analysis, one can impose assumptions on the sequence of page requests to more accurately model the behavioral characteristics of real-world page requests on computers. For example, page requests generally follow the \textit{locality of reference} principle, either in time (the next time a particular page is requested is close to its previous time) or in space (the page that will be requested next is likely to be in a nearby memory location).%

The Markov Paging model, introduced by \cite{karlin1992markov}, is one of the ways to model such locality of reference in the page requests. The model assumes that page requests are drawn from a (known) Markov chain over the entire set of $n$ pages. Once such a distributional assumption over the page requests is imposed, one can ask how an online paging algorithm fares, in comparison to the optimal online algorithm that suffers the smallest number of page misses \textit{in expectation} over the draw of a page sequence from the distribution. As it turns out, one can compute such an optimal online algorithm, but only in time exponential in the cache size $k$. Therefore, one seeks efficient polynomial-time algorithms that are approximately optimal.

In 1994, \cite{lund1994ip} proposed an extremely elegant randomized algorithm for the Markov Paging problem, known as the dominating distribution algorithm. The dominating distribution algorithm, whenever it suffers a cache miss, evicts a random page drawn from a special distribution, which has the property that a randomly drawn page is likely to be next requested latest among all the existing pages in the cache. This algorithm, which has since become a popular textbook algorithm due to its simplicity (e.g., see Chapter 5 in \cite{borodin2005online}, or the lecture notes by \cite{roughgardennotes}), runs in time polynomial in $k$; furthermore, \cite{lund1994ip} show that the algorithm suffers only 4 times more cache misses in expectation than the optimal online algorithm. Since then, this has remained state-of-the-art---we do not know any other polynomial-time algorithms that achieve a better performance guarantee. It has nevertheless been conjectured that the above guarantee for the dominating distribution algorithm is suboptimal (e.g., see Section 6.2 in \cite{roughgardennotes}).

Our main result in this work establishes an improved upper bound for the dominating distribution algorithm.

\begin{theorem}
    \label{thm:dominating-distribution-improved-bound}
    The dominating distribution algorithm suffers at most 2 times more cache misses in expectation compared to the optimal online algorithm in the Markov Paging model.%
\end{theorem}
In fact, as mentioned in \cite{lund1994ip}, our guarantee for the performance of the dominating distribution algorithm holds more generally for \textit{pairwise-predictive} distributions. These are distributions for which one can compute any time, for any pair of pages $p$ and $q$ in the cache, the probability that $p$ is next requested before $q$ in the future (conditioned on the most recent request).

While the dominating distribution algorithm is a randomized algorithm, \cite{lund1994ip} also propose a simple deterministic algorithm for paging, known as the median algorithm. On any cache miss, the median algorithm evicts the page in cache that has the largest median time at which it is requested next. \cite{lund1994ip} derive a guarantee for this algorithm in a slightly more restrictive setting than that of Markov paging; namely, the distribution of page requests is such that the inter-arrival times of different pages are independent. Under this assumption, \cite{lund1994ip} show that the median algorithm suffers at most 5 times more cache misses in expectation than the optimal online algorithm---to the best of our knowledge, this is the state of the art for deterministic algorithms in this setting. As a convenient consequence of our improved analysis, we can also improve this guarantee for the median algorithm.

\begin{theorem}
    \label{thm:median-improved-bound}
    The median algorithm suffers at most 4 times more cache misses in expectation compared to the optimal online algorithm in the distributional paging model with independent inter-arrival page request times.\footnote{The guarantee for the median algorithm in \cite{lund1994ip} only holds under the additional assumption of independent, inter-arrival page times. In a previous version of this manuscript, we mistakenly reported that the guarantee holds in the more general Markov paging model.}
\end{theorem}

Given the improved performance guarantee in \Cref{thm:dominating-distribution-improved-bound}, one can ask: is this the best bound that we can show for the dominating distribution algorithm? Or is there hope to show that it is \textit{the} optimal online algorithm? A detail here is that the dominating distribution algorithm is defined as any algorithm satisfying a specific property (see \eqref{eqn:dominating-distribution-property}); in particular, there can exist different dominating distribution algorithms. Our upper bound (\Cref{thm:dominating-distribution-improved-bound}) holds uniformly for \textit{all} dominating distribution algorithms. However, our next theorem shows a lower bound that at least one of these is significantly suboptimal. 

\begin{theorem}
    \label{thm:dominating-distribution-lower-bound}
    There exists a dominating distribution algorithm that suffers at least $1.5907$ times more cache misses in expectation compared to the optimal online algorithm in the Markov Paging model.
\end{theorem}

\Cref{thm:dominating-distribution-lower-bound} shows that one cannot hope to show optimality uniformly for all dominating distribution algorithms. To the best of our knowledge, no such lower bounds for the dominating distribution algorithm have been previously shown. While we believe that the bound in \Cref{thm:dominating-distribution-improved-bound} is the correct bound, it is an interesting open direction to close the gap between the upper and lower bound. %

Finally, we also consider a setting where the algorithm does not exactly know the Markov chain generating the page requests, but only has access to a dataset of past page requests. In this case, one can use standard estimators from the literature to approximate the transition matrix of the Markov chain from this dataset. But can an approximate estimate of the transition matrix translate to a tight competitive ratio for a paging algorithm, and if so, how large of a dataset does it require to achieve it? The robustness of the analysis\footnote{In fact, we borrow this robustness from the original analysis of \cite{lund1994ip}.} of \Cref{thm:dominating-distribution-improved-bound} lets us derive a precise learning-theoretic result (\Cref{thm:learning-markov-chains} in \Cref{sec:robustness}) showing that the dominating distribution algorithm, when trained on a dataset of size $O(n^2/\eps^2)$, is no more than $\frac{2}{1-2\eps}$ times worse than the optimal online algorithm on a fresh sequence of page requests. 

This result is especially relevant in the context of data-driven and learning-augmented algorithms. For instance, by treating the learned Markov chain as a learned oracle and combining it with a worst-case approach, such as the Marker algorithm, similarly to the approach proposed in~\citep{mahdian2012online,lykouris2021competitive}, the resulting algorithm achieves $O(1)$-consistency and $O(\log k)$-robustness guarantees for the online paging problem.

\begin{table}[H]
    \centering
    \setlength{\tabcolsep}{3pt}
    \renewcommand{\arraystretch}{0.95}
    \begin{tabular}{|c|c|c|c|c|}\hline
    Setting & Algorithm & Reference & Upper Bound & Lower Bound \\\hline
    \multirow{4}{*}{Markov Paging}
      & \multirow{4}{*}{\begin{tabular}{c}Dominating Distribution\\(Randomized)\end{tabular}}
      & \multirow{2}{*}{\cite{lund1994ip}} & 4 & \multirow{2}{*}{-} \\
      & &  & (Theorem 3.5) & \\ \cline{3-5}
      & & \multirow{2}{*}{This Work} & \textbf{2} & \textbf{1.5907} \\
      & & & (\Cref{thm:dominating-distribution-improved-bound}) & (\Cref{thm:dominating-distribution-lower-bound}) \\ \hline 
    \multirow{4}{*}{\begin{tabular}{c}Independent\\inter-arrival\\times\end{tabular}}
      & \multirow{4}{*}{\begin{tabular}{c}Median\\(Deterministic)\end{tabular}}
      & \multirow{2}{*}{\cite{lund1994ip}} & 5 & \textbf{1.511}\footnotemark  \\
      & &  & (Theorem 2.4) & (Theorem 5.4) \\ \cline{3-5}
      & & \multirow{2}{*}{This Work} & \textbf{4} & \multirow{2}{*}{-} \\
      & & & (\Cref{thm:median-improved-bound}) & \\ \hline
    \end{tabular}
    \caption{State of the art for Distributional Paging (best-known bounds are in bold). Both the median and the dominating distribution algorithm are due to \cite{lund1994ip}.}
    \label{table:summary}
\end{table}
\footnotetext{\cite{lund1994ip} only provide a proof for a lower bound of 1.4, but mention that they also obtained 1.511.}

\subsection{Other Related Work}
\label{sec:related-work}

There is a rich literature by this point on studying the online paging problem under assumptions on the permissible page requests that are supposed to model real-world scenarios (see e.g., the excellent surveys by \cite{irani2005competitive} and \cite{dorrigiv2005survey}). In particular, several models are motivated by the aforementioned locality of reference principle. The access graph model, proposed originally by \cite{borodin1991competitive}, and developed further in the works of \cite{irani1992strongly}, \cite{fiat1995randomized} and \cite{fiat1997truly}, assumes an underlying graph on the pages, which a request sequence is supposed to abide by. Namely, if a page $p_t$ has been requested at time $t$, then only the neighbors of $p_t$ in the graph may be requested at time $t+1$. In this sense, the Markov Paging model can be thought of as a probabilistic variant of the access graph model.
\cite{torng1998unified} models locality of reference based on the ``working sets'' concept, introduced by~\cite{denning1968working}, and also shows how a finite lookahead of the page sequence can be used to obtain improved guarantees. A related form of competitive analysis, proposed by \cite{koutsoupias2000beyond}, assumes that there is a family of valid distributions, and page sequences are drawn from some distribution belonging to this family. The goal of an algorithm is then to be simultaneously competitive with the optimal algorithm for every distribution in the family. A particular family of interest suggested by \cite{koutsoupias2000beyond} (which is incomparable to the Markov Paging model) is that of \textit{diffuse adversaries}, which has since been developed further in follow-up works by \cite{young1998bounding, young2000line} and \cite{becchetti2004modeling}. Finally, \cite{angelopoulos2009paging} study paging under the \textit{bijective analysis} framework, which compares the performance of two algorithms under bijections on page request sequences.

\section{Background and Setup}
\label{sec:background}

\subsection{Markov Paging}
\label{sec:markov-paging}

The Markov Paging model assumes that the page requests are drawn from a time-homogeneous Markov chain, whose state space is the set of all pages. More precisely, the Markov Paging model is specified by an initial distribution on the $n$ pages, and a transition matrix $M \in [0,1]^{n \times n}$, where $M_{i,j}$ specifies the probability that page $j$ is requested next, conditioned on the most recently requested page being $i$. The form of the initial distribution will not be too important for our purposes, but for concreteness, we can assume it to be the uniform distribution on the $n$ pages. It is typically assumed that the transition matrix $M$ is completely known to the paging algorithm/page eviction policy---however, the realized page sequence is not. Upon drawing $T$ page requests from the Markov chain, the objective function that a page eviction policy $\mcA$ seeks to minimize is:
\begin{align}
    \label{eqn:objective}
    \E[\cost(\mcA, T)] &= \E_{M, \mcA}\left[\sum_{t=1}^T\Ind[\text{$\mcA$ suffers a cache miss at time $t$}]\right],
\end{align}
where the expectation is with respect to both the random page request sequence drawn from $M$, and the randomness in the policy $\mcA$. We denote by $\opt$ the policy that minimizes the objective in \eqref{eqn:objective}. For $c \ge 1$, an eviction policy $\mcA$ is said to $c$-competitive against $\opt$ if for every $T$, $ \E[\cost(\mcA, T)] \le c \cdot \E[\cost(\opt, T)]$.

It is worth emphasizing again that the competitive analysis above is with respect to the optimal \textit{online} strategy, in that $\opt$ does not know beforehand the realized sequence of page requests. Contrast this to the best \textit{offline} strategy, which also additionally knows the realized sequence of page requests. As mentioned previously, the Farthest-in-Future policy \citep{belady1966study}, which at any cache miss, evicts the page that is next requested latest in the remainder of the sequence, is known to be an optimal offline policy. %

It is indeed possible to compute $\opt$ exactly---however, the best known ways do this are computationally very expensive, requiring an exponential amount of computation in $k$.

\begin{theorem}[Theorems 1, 2 in \cite{karlin1992markov}]
    \label{thm:opt-computation}
    The optimal online policy that minimizes $\lim_{T \to \infty}\frac{1}{T}\cdot\E[\cost(\mcA, T)]$ over all policies $\mcA$ can be computed exactly by solving a linear program in $n\binom{n}{k}$ variables. Furthermore, for any finite $T$, an optimal online policy that minimizes $\E[\cost(\mcA, T)]$ can be computed in time $T \cdot n^{\Omega(k)}$. %
\end{theorem}

In the latter case above when $T$ is finite, the optimal online policy can be computed, for example, via dynamic programming (using the ``Bellman operator''), where one writes the recurrence for taking a single step of the optimal policy, and unrolls time backwards from $T$. However, the state space in the dynamic program comprises of all $\{$cache state, requested page$\}$ pairs---there are $n^{\Omega(k)}$ many of these. Similarly, the linear program in the case where $T \to \infty$ also comprises of a single variable for each such state. Unfortunately, we do not know of any better algorithms that are both optimal and run in polynomial time in $k$. Therefore, one seeks efficient polynomial time (in $k$) algorithms that are approximately optimal, i.e., achieve $c$-competitiveness with $\opt$ for $c$ as close to $1$ as possible.

Towards this, \cite{karlin1992markov} showed that several intuitive eviction policies (e.g., on a cache miss, evict the page in cache that has the largest \textit{expected} time to be requested next) are necessarily at least an $\Omega(k)$ factor suboptimal compared to $\opt$. They then showed a policy $\mcA$ that is provably $O(1)$-competitive against $\opt$; however, the constant in the $O(1)$ is somewhat large. Thereafter, \cite{lund1999paging}\footnote{We reference the journal version \citep{lund1999paging} in lieu of the conference version \citep{lund1994ip} hereafter.} proposed a simple and elegant polynomial-time policy, referred to as the \textit{dominating distribution algorithm}, and showed that it achieves $4$-competitiveness against $\opt$. We describe this policy ahead.

\subsection{Dominating Distribution Algorithm}
\label{sec:dominating-distribution}

The dominating distribution algorithm, which was proposed by \cite{lund1999paging} and is denoted $\mcA_\dom$ hereon, operates as follows. At any page request, if $\mcA_\dom$ suffers a cache miss, it computes a so-called \textit{dominating} distribution $\mu$ over the currently-existing pages in the cache. Thereafter, it draws a page $p \sim \mu$, and evicts $p$. First, we will specify the properties of this dominating distribution. 

Intuitively, we want $\mu$ to be a distribution such that a \textit{typical} page drawn from $\mu$ will be highly likely to be next requested later than \textit{every} other page in the cache. Formally, suppose that the most recently requested page (that caused a cache miss) is $s$. Fix any two pages $p$ and $q$ (not necessarily in the cache), and let
\begin{align}
    \label{eqn:alpha-q-p-def}
    \alpha(p<q|s) := \Pr_M[\text{$p$ is next requested before $q$ }|\text{ $s$ is the most recently requested page}],
\end{align}
and let $\alpha(p<p|s)=0$ for all pages $p$. Here, the probability is only over the draw of page requests from the Markov chain.
The claim is that these $\alpha(p<q|s)$ values can be computed efficiently, in total time $O(n^4)$ for all pairs $p,q$ in cache, provided that we know the transition matrix $M$. To see this, let $x_s=\alpha(p<q|s)$. Note that $x_p=1$, $x_q=0$. For any other page $r \neq p,q$, note that
\begin{align*}
    x_r = M_{r,1}x_1 + \dots + M_{r,n}x_n.
\end{align*}
Thus, upon solving the linear system given by
\begin{align}
    \label{eqn:linear-system-for-alpha}
    \begin{bmatrix}
        M_{1,1} - 1 & M_{1,2} & \dots & M_{1,n} \\
        M_{2,1} & M_{2,2} -1 & \dots & M_{1,n} \\
        \vdots & \vdots & \cdots & \vdots\\
        \boldsymbol{e_p} \\
        \vdots & \vdots & \cdots &\vdots \\
        \boldsymbol{e_q} \\
        \vdots & \vdots & \cdots & \vdots \\
        M_{n,1} & M_{n,2} & \dots & M_{n,n}-1
    \end{bmatrix}
    \begin{bmatrix}
        x_1 \\ x_2 \\ \vdots \\ x_p \\ \vdots \\ x_q \\ \vdots \\ x_n
    \end{bmatrix} = 
    \begin{bmatrix}
        0 \\ 0 \\ \vdots \\ 1 \\ \vdots \\ 0 \\ \vdots \\ 0
    \end{bmatrix},
\end{align}
where $\boldsymbol{e_p}$ and $\boldsymbol{e_q}$ are row vectors with a $1$ in the $p^{\text{th}}$ and $q^{\text{th}}$ coordinates, respectively, and $0$ elsewhere, $x_s$ is the desired value $\alpha(p<q|s)$. %
In this way, we can solve a linear system for each of the $\le n^2$ possible pairs for $q,p$, and compute all $\alpha(p<q|s)$ values in time $O(n^4)$. For notational convenience, we will hereafter frequently denote $\alpha(p<q|s)$ simply by $\alpha(p<q)$, implicitly assuming conditioning on the most frequently requested page.

Equipped with all these $\alpha(p<q)$ values, let us define $\mu$ to be a distribution over the pages in the cache satisfying the following property: for every fixed page $q$ in the cache,
\begin{align}
    \label{eqn:dominating-distribution-property}
    \E_{p \sim \mu}[\alpha(p<q)] \le \frac{1}{2}.
\end{align}
This is the required dominating distribution. What may be somewhat surprising here is that such a dominating distribution, which satisfies the condition in \eqref{eqn:dominating-distribution-property} provably \textit{always} exists, and furthermore can be constructed efficiently by solving a linear program with just $O(k)$ variables (i.e., in $\poly(k)$ time).
\begin{theorem}[Theorem 3.4 in \cite{lund1999paging}]
    \label{thm:dominating-distribution-exists}
    A distribution $\mu$ satisfying the condition in \eqref{eqn:dominating-distribution-property} necessarily exists and can be computed by solving a linear program in $O(k)$ variables.
\end{theorem}
Thus, the total computation required by $\mcA_\dom$ at each cache miss is only $\poly(k)$ (to solve the linear program that constructs $\mu$), assuming all $\alpha(p<q)$ values are precomputed initially. In summary, the overall computation cost of $\mcA_\dom$ across $T$ page requests is at most $T \cdot \poly(n, k)$.

It remains to argue that if $\mcA_\dom$ evicts $p \sim \mu$ whenever it suffers a cache miss, then it fares well in comparison to $\opt$. A convenient observation towards showing this is the following claim, whose proof follows from the definition of $\mu$:
\begin{claim}
    \label{claim:dominating-distribution-probability}
    Let $s$ be the most recently requested page. The dominating distribution $\mu$ satisfies the following property: for every fixed page $q$ in the cache, we have that
    \begin{align}
        \Pr_{p \sim \mu, M}[\text{$q$ is next requested no later than $p$ }|\text{ $s$ is the most recently requested page}] \ge \frac{1}{2}.
    \end{align}
\end{claim}
\begin{proof}
    Writing out the condition \eqref{eqn:dominating-distribution-property} that $\mu$ satisfies more explicitly, we have that for every fixed page $q$ in cache,
    \begin{align*}
        &\sum_{\text{$p$ in cache}}\mu(p)\cdot \alpha(p<q|s) \le \frac{1}{2} \\
        \implies \quad & \sum_{\text{$p$ in cache}}\mu(p)\cdot \Pr_M[\text{$p$ is next requested before $q$ }|\text{ $s$ is the most recently requested page}] \le \frac{1}{2} \\
        \implies \quad &\Pr_{p \sim \mu, M}[\text{$p$ is next requested before $q$ }|\text{ $s$ is the most recently requested page}] \le \frac12 \\
        \implies \quad &\Pr_{p \sim \mu, M}[\text{$q$ is next requested no later than $p$ }|\text{ $s$ is the most recently requested page}] \ge \frac{1}{2}.
    \end{align*}
\end{proof}

Using this property of $\mu$, \cite{lund1999paging} show that $\mcA_\dom$ incurs only $4$ times as many cache misses as $\opt$ in expectation. The following section provides this part of their result.

\section{4-competitive Analysis of \cite{lund1999paging}}
\label{sec:lpr}

We restate and provide a more detailed proof of the crucial technical lemma from \cite{lund1999paging}, which is particularly useful for our improved $2$-competitive analysis in Section~\ref{sec:tight-analysis}.

\begin{lemma}[Lemma 2.5 in \cite{lund1999paging}]
    \label{lem:lpr}
    Let $\mcA$ be a paging algorithm. Suppose that whenever $\mcA$ suffers a cache miss, the page $p$ that it chooses to evict is chosen from a distribution such that the following property is satisfied: for every page $q$ in the cache, the probability (over the choice of $p$ and the random sequence ahead, conditioned on the most recently requested page) that $q$ is next requested no later than the next request for $p$ is at least $1/c$. Then $\mcA$ is $2c$-competitive against $\opt$.
\end{lemma}
\begin{proof}
    Consider an infinite (random) sequence of page requests, and consider running $\mcA$ and $\opt$ (independently of each other) on this sequence. At any time $t$ that $\mcA$ suffers a cache miss (say on a request to page $s$), suppose that $\mcA$ chooses to evict page $p$. Let $\mcA^-$ denote the contents in the cache of $\mcA$ just \textit{before} the request. Similarly, let $\opt^+$ denote the contents in the cache of $\opt$ just \textit{after} the request.
    
The proof uses a carefully designed charging scheme:\footnote{The charging scheme is purely for analysis purposes.} if $\mcA$ has to evict a page $p$ at time $t$ (due to a request for page $s$), $p$ assigns a charge to a page $c(p)$ that is \textit{not} in the cache of $\opt$ after the request at time $t$. Ideally, this page $c(p)$ will be requested again before $p$, causing $\opt$ to incur a cache miss at that later time. The specific charging scheme is as follows:
    \begin{figure}[H]
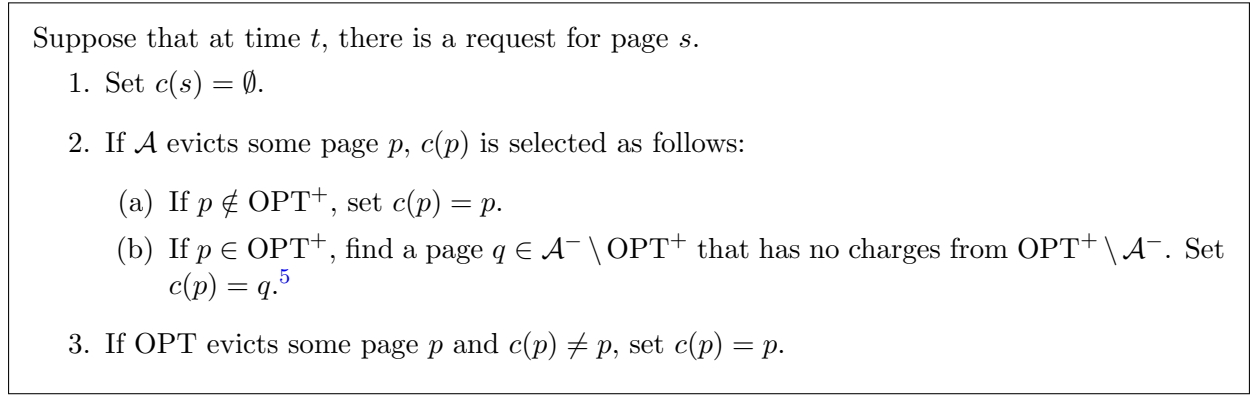

        \begin{framed}
        Suppose that at time $t$, there is a request for page $s$.
            \vspace{-2mm}
                \begin{enumerate}
                    \item \label{item:clear-charges} Set $c(s)=\emptyset$.
                    \item \label{item:assign-charge} If $\mcA$ evicts some page $p$, $c(p)$ is selected as follows:
                    \begin{enumerate}
                        \item \label{item:p-not-in-opt+} If $p \notin \opt^+$, set $c(p)=p$.
                        \item \label{item:p-in-opt+} If $p \in \opt^+$, find a page $q \in \mcA^- \setminus \opt^+$ that has no charges from $\opt^+ \setminus \mcA^-$. Set $c(p)=q$.\footnotemark
                    \end{enumerate}
                    \item \label{item:opt-eviction-charge-reassign} If $\opt$ evicts some page $p$ and $c(p) \neq p$, set $c(p)=p$.
                \end{enumerate}
            \vspace{-2mm}
        \end{framed}
        \vspace{-4mm}
        \caption{Charging Scheme}
        \label{fig:charging-scheme}
    \end{figure}
    \footnotetext{for any page $p$ that $\mcA$ might evict which satisfies the condition $p \notin \opt^+$, we make the \textit{same} choice of $q$.}
    At any time for any page $p$, $c(p)$ will either be $\emptyset$, or some (single) page. Furthermore, before the request at any time, $c(p) \neq \emptyset$ if and only if $p$ is not currently in the cache of $\mcA$.
    
    The first action (\Cref{item:clear-charges}) at the time of any page request is the clearing of charges: when a page $s$ is requested, if $s$ was giving a charge (either to itself or another page), it is duly cleared. If $\mcA$ suffers a cache miss, and evicts $p$, then our task is to find a page which is not going to be in the cache of $\opt$ just after time $t$, and assign a charge from $p$ to this page.
    
    If $p$ is not in $\opt^+$ (\Cref{item:p-not-in-opt+}), we assign $p$'s charge to itself. Otherwise, $p$ is in $\opt^+$ (i.e., $p$ is a page in $\mcA^- \cap \opt^+$), and we need to find a different page to charge. In this case, we will argue the existence of a page $q$ that satisfies the condition in \Cref{item:p-in-opt+}, for $p$ that fall into this case. Let us remove the common set $\mcA^- \cap \opt^+$ from each of $\mcA^-$ and $\opt^+$, and consider the sets $\mcA^- \setminus \opt^+$, $\opt^+ \setminus \mcA^-$. Note that both these sets have the same size, and $s$ is in $\opt^+ \setminus \mcA^-$, but not in $\mcA^- \setminus \opt^+$. Furthermore, in our initial step, we set $c(s)=\emptyset$. Thus, even if every page in $\opt^+ \setminus \mcA^-$ other than $s$ assigns a charge to a page in $\mcA^- \setminus \opt^+$, at least one page remains unassigned.
    We select this page (say the first one in a natural ordering) as $q$, the required charge recipient for $p$. Either way, note that $c(p)$ is not in $\opt^+$, and if $c(p)$ is requested before $p$, then $\opt$ incurs a cache miss. In this sense, $c(p)$ can be thought of as the “savior” page for $p$. %
    
    \Cref{item:opt-eviction-charge-reassign} simply reassigns the charge that a page gives if it is eventually evicted from $\opt$. 
    Specifically, $p$ may have initially set $c(p) = q \neq p$ when $\mcA$ evicted $p$, as $p$ still remained in $\opt$'s cache. However, if $\opt$ later evicts $p$, we can safely transfer the responsibility back from $q$ to $p$.

    The nice property about this charging scheme is that no page ever has more than $2$ charges on it at any time. Specifically, a page can hold at most one charge from itself and at most one charge from another page (because of \Cref{item:p-in-opt+} and \Cref{item:opt-eviction-charge-reassign}). 
    In fact, the only way a page can have two charges is if it first receives a charge from another page and later gets evicted by $\mcA$, assigning itself a self-charge.

    Now, fix a finite time horizon $T$: we will reason about the number of cache misses suffered by $\mcA$ and $\opt$ over the course of the $T$ page requests at $t=1,\dots,T$.
    First, let $\mcI$ denote the random variable that is the number of pages $s$ that were not present in the cache initially\footnote{Note that both $\opt$ and $\mcA$ start with the same initial cache.}, but were requested at some $t \le T$: $\opt$ suffers a cache miss for each of these pages.

    Next, let $r(t,p)$ denote the first time after $t$ that there is a request for page $p$. Define the indicator random variables:
    \begin{align}
        \alpha(t) &= \Ind[\text{$\mcA$ is forced to evict at time $t$}], \label{eqn:alpha}\\
        \beta(t) &= \Ind[\text{$\mcA$ is forced to evict at time $t$ and the page $\bp$ it evicts satisfies that $c(\bp)$\footnotemark\ is requested} \nonumber \label{eqn:beta}\\
        &\qquad\qquad\text{no later than $r(t,\bp)$}],
    \end{align}\footnotetext{While it is possible for charges to be reassigned, it suffices for the analysis to consider the indicator with $c(p)$ being the page that $p$ charges at the time of its eviction.}
    where $\bp$ is a random variable denoting the evicted page. 
    Note that
    \begin{align*}
        \cost(\mcA, T) &= \sum_{t \le T}\alpha(t).
    \end{align*}
    Observe that the indicator $\beta(t)$ represents a cache miss for $\opt$ (in the infinite sequence of requests) due to a request for the savior page $c(\bp)$. Note, however, that a request to the savior page $c(\bp)$ might occur \textit{after} the time horizon $T$. This can be addressed by counting the number of \textit{open} charges remaining after the request at time $T$. Specifically, let $\mcO$ denote the number of charges assigned by pages that are not in $\mcA$'s cache after the request at time $T$.
    Upon subtraction of $\mcO$, the quantity
    \begin{align*}
        \sum_{t \le T}\beta(t) - \mcO
    \end{align*}
    counts the number of cache misses suffered by $\opt$ due to requests to savior pages. A final detail is that a savior page could potentially have two charges on it, and hence we may count the same cache miss for $\opt$ twice in the above calculation. Concretely, suppose at time $t_1$, $p$ gets evicted by $\mcA$, and assigns a charge to $c(p)=q \neq p$. Then, suppose that at time $t_2 > t_1$ but $t_2 < r(t_1,p)$, $q$ is itself evicted by $\mcA$, resulting also in a charge $c(q)=q$. Now, suppose that $q$ gets requested after $t_2$ but before $r(t_1,p)$. In this case, we have that both $\beta(t_1)=1$ and $\beta(t_2)=1$, but these are really the same cache miss in $\opt$. Thus, we need to account for a possible double-counting of cache misses in $\opt$---one way to (very conservatively) do this is simply dividing the expression by 2. In total, we obtain that
    \begin{align}
        \cost(\opt, T) &\ge \frac{1}{2}\left(\sum_{t \le T}\beta(t) - \mcO \right) + \mcI, \label{eqn:lpr-final-accounting}
    \end{align}
    and hence
    \begin{align*}
        \E[\cost(\opt, T)] &\ge \frac{1}{2}\sum_{t \le T}\E[\beta(t)] + \E\left[\mcI - \frac{\mcO}{2}\right] \\
        &\ge \frac{1}{2}\sum_{t \le T}\E[\beta(t)] + \E\left[\mcI - \mcO\right].
    \end{align*}
    The expectation on the left side above is only with respect to the randomness in the sequence of page requests and $\opt$, whereas the expectation on the right side is over the randomness in the sequence of requests, $\opt$ as well as $\mcA$. 
    
    Now, fix any $t \le T$. Let $\sigma_{\le t}$ denote a fixed sequence of $t$ page requests, $\opt_{\le t}$ denote the execution of $\opt$ on these $t$ page requests, and $\mcA_{<t}$ denote the execution of $\mcA$ on all but the last of these page requests, such that $\sigma_{\le t}, \opt_{\le t}, \mcA_{< t}$ together result in a cache miss for $\mcA$ at time $t$. Then, we have that
    \begin{align*}
        \E[\beta(t)] &= \Pr[\text{$\mcA$ is forced to evict at time $t$ and the page $\bp$ it evicts satisfies that $c(\bp)$ is requested} \nonumber \\
        &\text{no later than $r(t,\bp)$}] \\
        &\hspace{-0.8cm}= \sum_{\sigma_{\le t}, \opt_{\le t}, \mcA_{<t}} \Pr[\sigma_{\le t}, \opt_{\le t}, \mcA_{<t}] \cdot \Pr\left[\text{page $\bp$ evicted by $\mcA$ at time $t$ satisfies that $c(\bp)$} \right. \\ &\qquad\qquad\qquad\qquad\qquad\qquad\qquad\qquad
        \left.\text{is requested no later than $r(t,\bp)$}\ \big|\ \sigma_{\le t}, \opt_{\le t}, \mcA_{<t}\right],
    \end{align*}
    Let us use the shorthand $\mcB:=\sigma_{\le t}, \opt_{\le t}, \mcA_{<t}$, and focus on the latter term in the summation.
    Note that once we have conditioned on $\mcB$, the configuration of $\mcA$'s cache (before the page request at $t$) is determined---denote its pages by $\mcA^-$.  At the last page request in $\sigma_{\le t}$, $\mcA$ suffers a cache miss, and chooses a page to evict from $\mcA^-$ from a (conditional) distribution $\mcP_t$ satisfying the property given in the lemma statement. Namely,
    \begin{align*}
        &\Pr\left[\text{page $\bp$ evicted by $\mcA$ at time $t$ satisfies that $c(\bp)$ is requested no later than $r(t,\bp)$}\ \big|\ \mcB\right] \\
        &= \sum_{p \in \mcA^-}\Pr_{\mcP_t}[p] \Pr\left[\text{$c(p)$ is requested no later than $r(t,p)$}\ \big|\ \mcB\right].
    \end{align*}
    But note that conditioning on $\sigma_{\le t}$, $\opt_{\le t}$ determines the cache of $\opt$ after time $t$: let its contents be denoted by $\opt^+$. Then, according to our charging scheme, for any $p \in \mcA^-$, $c(p)$ is as follows: if $p \notin \opt^+, c(p)=p$, whereas if $p \in \opt^+$, $c(p)=q$ for some fixed $q \in \mcA^-\setminus \opt^+$ that satisfies the condition in \Cref{item:p-in-opt+}. Note importantly that we make the same choice of $q$ for any $p$ in the latter case.
    \begingroup
    \allowdisplaybreaks
    \begin{align*}
        &\sum_{p \in \mcA^-}\Pr_{\mcP_t}[p] \Pr\left[\text{$c(p)$ is requested no later than $r(t,p)$}\ \big|\ \mcB\right] \\
        &= \sum_{p \in \mcA^-, c(p)=p}\Pr_{\mcP_t}[p]\underbrace{ \Pr\left[\text{$p$ is requested no later than $r(t,p)$}\ \big|\ \mcB\right]}_{=1\text{ since $p$ itself is the page requested at $r(t,p)$}} \\
        &+ \sum_{p \in \mcA^-, c(p)=q}\Pr_{\mcP_t}[p] \Pr\left[\text{$q$ is requested no later than $r(t,p)$}\ \big|\ \mcB\right] \\
        &\ge \sum_{p \in \mcA^-, c(p)=p}\Pr_{\mcP_t}[p] \Pr\left[\text{$q$ is requested no later than $r(t,p)$}\ \big|\ \mcB\right] \\
        &+ \sum_{p \in \mcA^-, c(p)=q}\Pr_{\mcP_t}[p] \Pr\left[\text{$q$ is requested no later than $r(t,p)$}\ \big|\ \mcB\right] \\
        &= \sum_{p \in \mcA^-}\Pr_{\mcP_t}[p] \Pr\left[\text{$q$ is requested no later than $r(t,p)$}\ \big|\ \mcB\right] \\
        &\ge \frac{1}{c},
    \end{align*}%
    \endgroup
    where in the last line, we used the property of the distribution $\mcP_t$ from the lemma statement. %
    Tracing backwards, we have obtained that
    \begin{align*}
        \E[\beta(t)] &\ge \sum_{\sigma_{\le t}, \opt_{\le t}, \mcA_{<t}} \Pr[\sigma_{\le t}, \opt_{\le t}, \mcA_{<t}] \cdot \frac{1}{c} = \frac{1}{c}\cdot \Pr[\text{$\mcA$ is forced to evict at time $t$}] = \frac{1}{c} \cdot \E[\alpha(t)].
    \end{align*}
    Finally, we observe that the random variable $\mcI-\mcO$ is always nonnegative. This is because, any page that is in the cache of $\mcA$ after the request at time $T$, or not in the cache of $\mcA$ after time $T$ but giving a charge, \textit{must necessarily} have either been in the initial cache, or must have been requested at some time $t \le T$. This implies that $k+\mcO \le k + \mcI$, which implies that $\mcI-\mcO \ge 0$. In total, 
    \begin{align*}
        \E[\cost(\opt, T)] \ge \frac12\sum_{t \le T}\E[\beta(t)] + \E\left[\mcI - \mcO\right]
        \ge \frac{1}{2c}\E\left[\sum_{t \le T}\alpha(t)\right] = \frac{1}{2c}\E[\cost(\mcA, T)]
    \end{align*}
    as required.    
\end{proof}

\begin{corollary}
    \label{corollary:dom-is-4-competitive}
    The dominating distribution algorithm $\mcA_\dom$ is 4-competitive against $\opt$.
\end{corollary}
\begin{proof}
    This follows from \Cref{lem:lpr} and \Cref{claim:dominating-distribution-probability}.
\end{proof}

\begin{remark}
    \label{remark:suboptimal-fifo}
    Consider setting $c=1$ in \Cref{lem:lpr}: this corresponds to $\mcA$ effectively being an \textit{offline} algorithm, which \textit{knows} which page in the cache is going to next be requested farthest in the future compared to all the other pages in the cache. We already know that such an algorithm is optimal, i.e., is 1-competitive against $\opt$. However, the guarantee given by \Cref{lem:lpr} for such an algorithm is still only $2$-competitiveness. This at least suggests that there is scope to improve the guarantee to $c$-competitiveness. %
\end{remark}

\section{Sources of Looseness in the Analysis}
\label{sec:looseness}
We systematically identify the sources of looseness in the above analysis with illustrative examples, before proceeding to individually tighten them in the subsequent section.

\subsection{A Conservative Approach to Handling Doubly-Charged Pages}
\label{sec:looseness-double-counting}

Recall that in the summation $\sum_{t \le T} \beta(t)$ above, a cache miss for $\opt$ may be double-counted if the same page is assigned two charges. Although one way to address this overcounting is to divide the summation by 2, this approach is overly conservative, as it \textit{under}counts every cache miss in $\opt$ that results from a request to a singly-charged page. For instance, consider the following situation:
\begin{figure}[H]
    \centering
    \includegraphics[scale=0.5]{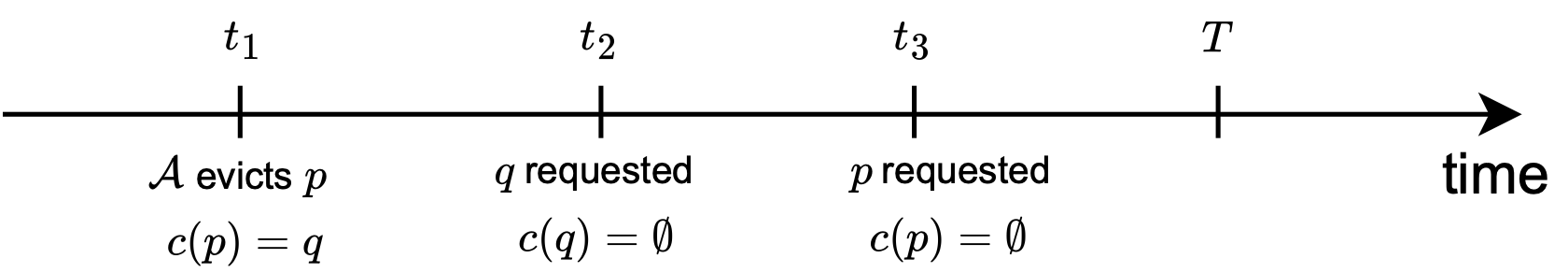}
    \caption{Request for $q$ (a singly charged page) at $t_2$ constitutes a unit cache miss for $\opt$, but the division by 2 undercounts this (and every other singly-charged cache miss in $\opt$).}
    \label{fig:double-counting}
\end{figure}
Here, at time $t_1$, $\mcA$ evicts a page $p$. Suppose that $p$ continues to live in the cache of $\opt$ after the request at $t_1$, and hence $c(p)$ is assigned to some $q$ in $\mcA$'s cache. Now, say that $q$ is requested at time $t_2$, which is before $t_3=r(t_1, p)$. Furthermore, assume that at $t_2$, $q$ only had the single charge on itself (by $p$), i.e., $q$ was not evicted by $\mcA$ in the time between $t_1$ and $t_2$. Then, $\beta(t_1)=1$, but we are wastefully dividing it by 2 in our calculation.

\subsection{Inexhaustive Clearing of Charges upon a Page Request}
\label{sec:looseness-uncleared-charges}

Consider the first item in the charging scheme (\Cref{item:clear-charges})---whenever a page $s$ is requested, any charges that $s$ might be \textit{giving} are cleared (i.e., $c(s)=\emptyset$). Intuitively, this is supposed to account for the fact that, while $s$ was holding some page $c(s)$ responsible for being requested before $s$ itself, either this did happen, in which case we happily clear the already-paid charge, or this did not quite happen and $s$ got requested before $c(s)$, in which case we should let go of the charge and focus on the future. However, consider instead the case where $s$ does not have a charge on any other page, but is itself the bearer of a charge by some other page, say $p$. In this case, $s$ has successfully paid the charge that was expected of it---\textit{but this charge would only be cleared upon the next request to $p$!} If it so happens that $p$ is next requested \textit{after} the time horizon $T$, then even if $s$ successfully paid the charge due to $p$, since this charge was not cleared when it was requested, it would be counted in $\mcO$ as part of the open charges post time $T$, and wastefully subtracted in the calculation. This is concretely illustrated in the situation below:
\begin{figure}[H]
    \centering
    \includegraphics[scale=0.45]{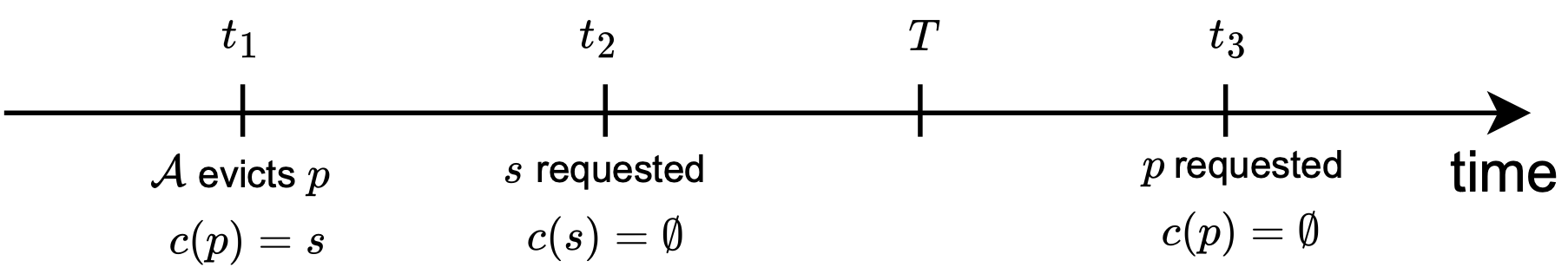}
    \caption{Here, $\beta(t_1)=1$, and $\opt$ suffers a cache miss at $t_2$. However, $p$ still holds a charge on $s$ at time $T$ because it is requested at $t_3 > T$. Thus, this charge is included in $\mcO$, canceling out the contribution due to the cache miss at $t_2$.}
    \label{fig:uncleared-charges}
\end{figure}
This suggests that whenever a page is requested, we should clear not only the charges it gives but also any charges it {\em bears}, preventing unnecessary inclusion in $\mcO$. This approach ensures that a page discards its imposed charge immediately upon ``paying it off" rather than with a delay.

\subsection{No Accounting for Uncharged, Non-first-timer Pages}
\label{sec:looseness-uncharged-non-first-time}

Finally, observe that in the accounting of cache misses for $\opt$, we only count those that occur due to requests to charged/savior pages (as counted by the $\beta(t)$'s), and those that occur due to first-time requests to pages not initially in the cache (as counted by $\mcI$). However, $\opt$ can also suffer cache misses due to a request to an \textit{uncharged} page. Namely, if there is a request to a page $q$ that was previously evicted by $\opt$, but at the time that it is requested, $q$ is bearing no charges at all, then the cache miss due to $q$ is not being counted in the calculation. This is illustrated by the scenario below: 
\begin{figure}[H]
    \centering
    \includegraphics[scale=0.45]{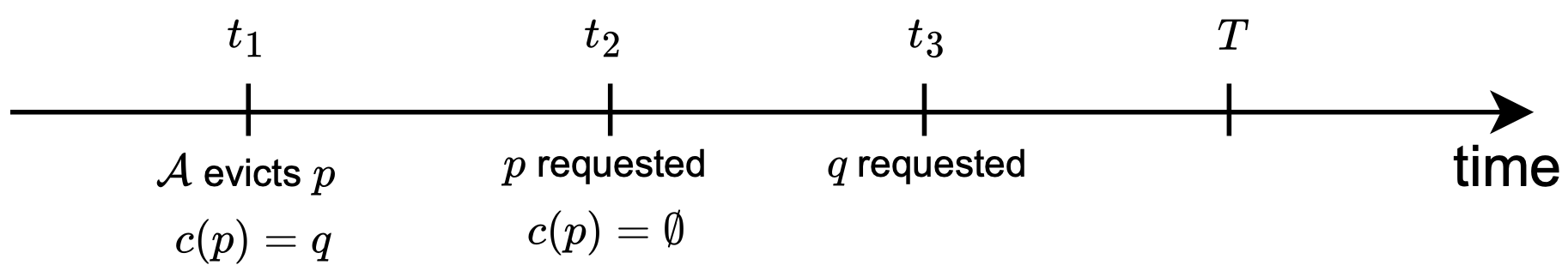}
    \caption{At $t_3$, $q$ does not have any charges on it, but still causes a cache miss for $\opt$.}
    \label{fig:uncharged-eviction}
\end{figure}
$\mcA$ evicts $p$ at time $t_1$ and assigns a charge to $q \neq p$, implying that $q$ is not in the cache of $\opt$ after the request at $t_1$ (but was instead previously evicted by $\opt$). Next, at $t_2$, $p$ is requested, and this clears the charge it had on $q$. Since $p$ got requested before $q$, $\beta(t_1)=0$. Then, at $t_3$, $q$ is requested---$q$ has no charges on it at this point. Notice that this still causes a cache miss for $\opt$ at $t_3$. However, this cache miss is not accounted for in our calculation, either by $\beta(t_1)$, or by $\mcI$ (since this is not the first time that $q$ is being brought into the cache).

\section{Tightening the Analysis to 2-competitiveness}
\label{sec:tight-analysis}
Having identified the loose ends in the analysis of \cite{lund1999paging} above, we are able to tighten their analysis and prove the following lemma:
\begin{lemma}
    \label{lem:tight-analysis}
    Let $\mcA$ be a paging algorithm. Suppose that whenever $\mcA$ suffers a cache miss, the page $p$ that it chooses to evict is chosen from a distribution such that the following property is satisfied: for every page $q$ in the cache, the probability (over the choice of $p$ and the random sequence ahead, conditioned on the most recently requested page) that $q$ is next requested no later than the next request for $p$ is at least $1/c$. Then $\mcA$ is $c$-competitive against $\opt$.
\end{lemma}
\begin{proof}
    We make one small change (highlighted in green) to \Cref{item:clear-charges} in the charging scheme (\Cref{fig:charging-scheme}) from the analysis of \cite{lund1999paging}---whenever a new page is requested, any charges that this page might be giving are cleared, but also any charges that other pages might be having on the requested page are also cleared. This fixes the issue regarding uncleared charges discussed in \Cref{sec:looseness-uncleared-charges}.
    \begin{figure}[H]
        \begin{framed}
        Suppose that at time $t$, there is a request for page $s$.
            \vspace{-2mm}
                \begin{enumerate}
                    \item \label{item:updated-clear-charges} \begin{enumerate}
                        \item \label{item:updated-clear-charges-giving} Set $c(s)=\emptyset$.
                        \item \label{item:updated-clear-charges-bearing} \green{For any page $p$ that has $c(p)=s$, set $c(p)=\emptyset$.}
                    \end{enumerate}
                    \item \label{item:updated-assign-charge} If $\mcA$ evicts some page $p$, $c(p)$ is selected as follows:
                    \begin{enumerate}
                        \item \label{item:updated-p-not-in-opt+} If $p \notin \opt^+$, set $c(p)=p$.
                        \item \label{item:updated-p-in-opt+} If $p \in \opt^+$, find a page $q \in \mcA^- \setminus \opt^+$ that has no charges from $\opt^+ \setminus \mcA^-$. Set $c(p)=q$.
                    \end{enumerate}
                    \item \label{item:updated-opt-eviction-charge-reassign} If $\opt$ evicts some page $p$ and $c(p) \neq p$, set $c(p)=p$.
                \end{enumerate}
            \vspace{-2mm}
        \end{framed}
        \vspace{-4mm}
        \caption{Updated Charging Scheme}
        \label{fig:updated-charging-scheme}
    \end{figure}
    Let $\alpha(t)$ and $\beta(t)$ be the same random variables as defined in \eqref{eqn:alpha}, \eqref{eqn:beta} in the proof of \Cref{lem:lpr}. Further, let $\mcI$ and $\mcO$ be the same random variables as defined there as well. Note that even with the slightly-changed charging scheme, $\mcI$ remains quantitatively the same; the random variable $\mcO$ however changes---in particular, it potentially becomes a smaller number, because we are clearing charges more aggressively. The random variable
    \begin{align*}
        \sum_{t \le T}\beta(t) - \mcO
    \end{align*}
    nevertheless still counts the cache misses suffered by $\opt$ due to requests to savior pages, while potentially double-counting a few misses. However, to account for this double-counting, instead of dividing every miss by 2 (which was the issue discussed in \Cref{sec:looseness-double-counting}), we explicitly keep track if some $\beta(t)$ corresponds to a request to a doubly-charged page, and subtract it from our calculation. Formally, let
    \begin{align}
        \mcD = \sum_{t \le T}\Ind[\text{Page $\bp$ requested at time $t$ has two charges on it}]. \label{eqn:D}
    \end{align}
    Then, the quantity
    \begin{align*}
        \sum_{t \le T}\beta(t) - \mcD - \mcO
    \end{align*}
    counts the cache misses suffered by $\opt$ due to requests to savior pages more precisely, and without double-counting any miss. This crucially allows us to avoid an unnecessary factor of $2$.

    Lastly, in order for the analysis to go through, we have to also fix the final issue regarding requests to uncharged pages described in \Cref{sec:looseness-uncharged-non-first-time}. We simply do this by keeping track of an additional quantity that counts this, and add it to our calculation. Namely, let
    \begin{align}
        \mcU = \sum_{t \le T}\Ind[&\text{Page $\bp$ requested at time $t$ does not exist in $\opt$'s cache at this time because it was} \nonumber \\ 
        &\text{previously evicted by $\opt$, and $\bp$ has no charges on it}]. \label{eqn:U}
    \end{align}
    Then, combining all of the above, we have that
    \begin{align}
        \cost(\opt, T) &\ge \sum_{t \le T}\beta(t) - \mcD - \mcO + \mcU. \label{eqn:updated-final-accounting}
    \end{align}
    Comparing this to the accounting in \eqref{eqn:lpr-final-accounting}, it at least seems plausible that by not halving, we might be able to save a factor of 2. Namely, taking expectations, we obtain that
    \begin{align*}
        \E[\cost(\opt, T)] &\ge \sum_{t \le T} \E[\beta(t)] + \E[\mcI - \mcD - \mcO + \mcU] \\
        &\ge \frac{1}{c}\sum_{t \le T}\E[\alpha(t)] + \E[\mcI - \mcD - \mcO + \mcU] \\
        &= \frac{1}{c}\E[\cost(\mcA, T)] + \E[\mcI - \mcD - \mcO + \mcU],
    \end{align*}
    where the first inequality follows from the same analysis that we did in the proof of \Cref{lem:lpr}. It remains to argue that the quantity $\E[\mcI - \mcD - \mcO + \mcU]$ is nonnegative. Note that the random variable $\mcI - \mcD - \mcO + \mcU$ starts out being $0$ before the very first request at time $t=1$. We will argue that it always stays nonnegative thereafter via the following two claims.

    \begin{claim}
        \label{claim:potential-decrease-characterization}
         If the page $s$ requested at time $t$ satisfies the following condition: $s$ exists in $\opt$'s cache and $s$ does not exist in $\mcA$'s cache and $s$ is not giving a charge and $s$ is not bearing any charges, then the random variable $\mcI-\mcD-\mcO+\mcU$ decreases by 1 from $t$ to $t+1$. Furthermore, if the page $s$ requested at time $t$ does not satisfy this condition, then $\mcI-\mcD-\mcO+\mcU$ either increases or stays the same from $t$ to $t+1$.
    \end{claim}
    \begin{proof}
        We first show that if the page $s$ requested at time $t$ satisfies the condition, then $\mcI - \mcD - \mcO + \mcU$ decreases by 1. Since $s$ is in $\opt$'s cache, the request to $s$ does not change $\mcU$ (see the definition in \eqref{eqn:U}). Moreover, $s$ being in $\opt$'s cache implies it was either in the initial cache or previously requested, so $\mcI$ also remains unchanged. Similarly, because $s$ is not doubly charged, $\mcD$ remains unchanged. Thus, we only need to account for the change in $\mcO$.

        Because $s$ is neither giving nor bearing any charges, both \Cref{item:updated-clear-charges-giving} and \Cref{item:updated-clear-charges-bearing} cause no change in $\mcO$. Finally, because $s$ does not exist in $\mcA$'s cache and results in a cache miss, \Cref{item:updated-assign-charge} creates a new charge, causing $\mcO$ to increase by 1. In total, $\mcI-\mcD-\mcO+\mcU$ decreases by 1.

        \smallskip
        Next, we will show that if $s$ does not satisfy the condition, then $\mcI-\mcD-\mcO+\mcU$ either increases or stays the same. Towards this, consider each of the following cases: 
        
        \medskip\noindent\underline{Case 1: $s$ is a doubly-charged page.} \\
        In this case, $\mcD$ increases by 1. Note that one of the charges on $s$ is by $s$ itself, and the other charge is by some other page $q$. In particular, $s$ is not \textit{giving} a charge to a page other than itself. Thus, when $s$ is requested, \Cref{item:updated-clear-charges-giving} and \Cref{item:updated-clear-charges-bearing} ensure that the charge by $s$ on itself as well as the charge on $s$ by $q$ are \textit{both} dropped.\footnote{Note how \Cref{item:updated-clear-charges-bearing} was necessary to ensure this.} However, note also that $s$ is not in the cache of $\mcA$ (since it has a charge on itself, it was previously evicted by $\mcA$), and hence this request has caused a cache miss in $\mcA$, resulting in the creation of a new charge in \Cref{item:updated-assign-charge}. \Cref{item:updated-opt-eviction-charge-reassign} can only reassign a charge, and thus, the net change in $\mcO$ is $-1$. Finally, $\mcI$ and $\mcU$ remain the same. Therefore, the overall change in $\mcI - \mcD - \mcO + \mcU$ is 0. \medskip
        
        \noindent\underline{Case 2: $s$ is a singly-charged page.} \\
        This means that either $s$ is not in the cache of $\mcA$ and is charging itself, or $s$ is in the cache of $\mcA$ and is bearing a charge given by some other page $p$ at the time of its eviction. 
        
        In the former case, observe that none of $\mcI$, $\mcD$ or $\mcU$ are affected. The clearing of charges in \Cref{item:updated-clear-charges} causes $\mcO$ to decrease by 1, and the cache miss causes the creation of a new charge in \Cref{item:updated-assign-charge}. In total, $\mcI-\mcD-\mcO+\mcU$ stays unchanged.
        In the latter case, observe that the charge on $s$ is cleared in \Cref{item:updated-clear-charges-bearing}, causing $\mcO$ to decrease by 1. Furthermore, because $s$ is in the cache of $\mcA$, there is no creation of a new charge in \Cref{item:updated-assign-charge}. Thus, $\mcI-\mcD-\mcO+\mcU$ increases by 1.

        \medskip\noindent\underline{Case 3: $s$ is not in $\opt$'s cache.} \\
        If $s$ is either singly or doubly-charged, we fall under Case 1 or 2. If $s$ has no charges on it, then:
        
        \smallskip         \noindent\underline{Subcase 3a: $s$ has never been requested before.} \\
        In this case, $\mcI$ increases by 1. Also, $s$ cannot be giving or bearing any charges; thus \Cref{item:updated-clear-charges-giving} and \Cref{item:updated-clear-charges-bearing} cause no change to $\mcO$. However, the request to $s$ causes a cache miss in $\mcA$, resulting in an eviction, and the creation of a new charge. Thus, $\mcO$ increases by 1. Additionally, $\mcD$ and $\mcU$ remain the same. Thus, the net change in $\mcI - \mcD - \mcO + \mcU$ is 0.
        
        \smallskip         \noindent\underline{Subcase 3b: $s$ has been requested before.} \\
        This means that $s$ was previously in $\opt$'s cache at some time but was since evicted. Also, $s$ has no charges on it. Thus, $\mcU$ increases by 1, while $\mcI$ and $\mcD$ remain unchanged. It remains to reason about $\mcO$. Because $s$ has no charges on it, \Cref{item:updated-clear-charges-bearing} causes no change in $\mcO$.
        \Cref{item:updated-clear-charges-giving} either decreases $\mcO$ by 1 or causes no change to it.
        While the request to $s$ can cause a cache miss to $\mcA$, this can only result in the creation of a single new charge, and $\mcO$ can increase by at most 1 due to this. In total, $\mcI-\mcD-\mcO+\mcU$ either increases %
        or stays the same. 
        
        \medskip
        \noindent\underline{Case 4: $s$ is in $\mcA$'s cache.} \\
        This means that $s$ is giving no charges. Then, either $s$ is singly-charged or has no charges on it. It cannot be doubly-charged because it would have to be out of $\mcA$'s cache for that. If it is singly-charged, we fall under Case 2. If it has no charges on it, we reason as follows: both \Cref{item:updated-clear-charges} and \Cref{item:updated-assign-charge} leave $\mcO$ unaffected. Moreover, $\mcI$ and $\mcD$ remain unaffected. Finally, depending on whether $s$ is or isn't in $\opt$'s cache, $\mcU$ either stays the same or increases. Hence, $\mcI-\mcD-\mcO+\mcU$ either stays the same or increases. 
        
        \medskip
        \noindent\underline{Case 5: $s$ is giving a charge.} \\
        If $s$ has any charges on it, we fall under Case 1 or 2. So we assume that $s$ has no charges on it. This means that $s$ is not in $\mcA$'s cache, but is still in $\opt$'s cache. Note then that $\mcI, \mcD, \mcU$ remain unchanged. We reason about the change in $\mcO$ as follows: \Cref{item:updated-clear-charges-giving} clears the charge that $s$ is giving and decreases $\mcO$ by 1, whereas \Cref{item:updated-assign-charge} creates a new charge and increases $\mcO$ by 1. In total, $\mcI-\mcD-\mcO+\mcU$ remains unchanged.

        Thus, in any way that the requested page $s$ might not satisfy the condition, $\mcI-\mcD-\mcO+\mcU$ either increases or stays the same, concluding the proof.
    \end{proof}

    \begin{claim}
        If the page $s$ requested at time $t$ satisfies the condition: $s$ exists in $\opt$'s cache and $s$ does not exist in $\mcA$'s cache and $s$ is not giving a charge and $s$ is not bearing any charges, then $\mcI-\mcD-\mcO+\mcU$ is strictly positive just before this request.
    \end{claim}
    \begin{proof}
        Consider any time $t$ where the requested page $s$ satisfies the condition. We can associate to the page $s$ two past events occurring at times $t_1$ and $t_2$ (where $t_1 < t_2 < t$) such that: (1) $\mcA$ evicts $s$ at $t_1$, but $s$ continues to live in $\opt$'s cache, resulting in $s$ giving a charge to some other page $q$ in $\mcA$'s cache (\Cref{item:p-in-opt+}), and (2) $q$ gets requested at $t_2$, before the request to $s$ at $t$, thereby clearing the charge by $s$ on $q$ (as per \Cref{item:updated-clear-charges-bearing}). Observe that these two past events need to necessarily occur for $s$ to satisfy the condition. Now observe that when $q$ gets requested at $t_2$, $q$ is the bearer of a \textit{single} charge. Thus, over the course of this request to $q$, $\mcI, \mcD, \mcU$ remain unchanged. Because $s$'s charge on $q$ gets cleared, $\mcO$ decreases by $1$. Finally, because $q$ is in $\mcA$'s cache at this time, no new charges are added to $\mcO$. Thus, $\mcI - \mcD - \mcO + \mcU$ strictly \textit{increases} by 1 at step (2). Also, note that for each distinct time step where the requested page $s$ satisfies the condition, there is a distinct associated (past) step (2).

        Now consider the \textit{first} time $t_0$ where the requested page $s$ satisfies the condition. Then, by \Cref{claim:potential-decrease-characterization}, at every previous time step, $\mcI - \mcD - \mcO + \mcU$ either increased or remained the same. Given that $\mcI - \mcD - \mcO + \mcU$ started out being 0, and recalling that step (2) (which happened before $t_0$) caused a strict \textit{increase} in the quantity, we have that $\mcI - \mcD - \mcO + \mcU$ is strictly positive at $t_0$.

        Just after the request at $t_0$, $\mcI - \mcD - \mcO + \mcU$ decreases by 1 (by \Cref{claim:potential-decrease-characterization}), and is now only guaranteed to be nonnegative, instead of positive. Then, consider the next time $t$ when a page $s$ satisfying the condition is requested. We can again trace its associated (distinct) step (2) that happened in the past. If this happened before $t_0$, then $\mcI - \mcD - \mcO + \mcU$ stayed positive after the request at $t_0$. Alternatively, if this happened in between $t_0$ and $t$, $\mcI - \mcD - \mcO + \mcU$ still turns positive (if it ever became zero at all). In either case, $\mcI - \mcD - \mcO + \mcU$ is positive before the request at $t$. The claim follows by induction.
    \end{proof}

    The above two claims establish that if $\mcI - \mcD - \mcO + \mcU=0$ at any time $t$, then $\mcI - \mcD - \mcO + \mcU$ cannot decrease at time $t+1$. Furthermore, if at all $\mcI - \mcD - \mcO + \mcU$ does decrease (from being a positive number), it only decreases by 1. Together, this means that $\mcI - \mcD - \mcO + \mcU$ is always nonnegative, %
    and hence it is nonnegative in expectation. This concludes the proof of \Cref{lem:tight-analysis}
\end{proof}

\begin{corollary}
    \label{corollary:dom-is-2-competitive}
    The dominating distribution algorithm $\mcA_\dom$ is 2-competitive against $\opt$.
\end{corollary}
\begin{proof}
    This follows from \Cref{lem:tight-analysis} and \Cref{claim:dominating-distribution-probability}.
\end{proof}

\begin{remark}
    \label{remark:optimal-fifo}
    As in \Cref{remark:suboptimal-fifo}, if we set $c=1$ in \Cref{lem:tight-analysis}, our lemma says that a policy that has essentially seen the future and evicts the page that is next scheduled to be requested latest is 1-competitive against $\opt$, i.e., it is optimal. Thus, our analysis, while establishing the best-possible guarantee for algorithms satisfying the condition of the lemma, additionally recovers an alternate, charging-scheme-based proof of the optimality of the Farthest-in-Future eviction policy.
\end{remark}

\begin{remark}
    \label{remark:2-correct-answer-for-dom}
    The factor of 4 in the analysis of \cite{lund1994ip} constituted a factor of 2 arising due to doubly-charged pages, and a factor of 2 arising from the property of the dominating distribution algorithm. While we got rid of the first factor, the second factor appears to be necessary, stemming from an inherent property (\Cref{eqn:dominating-distribution-property}) of the algorithm itself. It would indeed be very interesting to see if this is not the case, and if the upper bound could further be improved from 2.
\end{remark}

\Cref{lem:tight-analysis} also allows us to improve the approximation guarantee for another intuitive \textit{deterministic} algorithm considered by \cite{lund1999paging}, namely the \textit{median} algorithm, in the setting where one additionally assumes that the inter-arrival times between consecutive requests for a page are independent across all the pages.  At any cache miss, the median algorithm evicts the page in cache that has the largest median time of next request. For deterministic algorithms satisfying the condition of \Cref{lem:lpr}, \cite{lund1999paging} employ a slightly more specialized analysis of the charging scheme in \Cref{fig:charging-scheme} (Lemma 2.3 in \cite{lund1999paging}) to obtain a $(c+1)$-competitive guarantee against $\opt$ (instead of $2c$-competitiveness). Thereafter, by arguing that the median algorithm satisfies the condition with $c=4$, (Theorem 2.4 in \cite{lund1999paging}), they are able to show that the median algorithm is 5-competitive against $\opt$. Our tighter analysis in \Cref{lem:tight-analysis} applies to any algorithm (deterministic/randomized), and hence also improves the guarantee for deterministic algorithms obtained by \cite{lund1999paging}. More importantly, it improves the performance guarantee that we can state for the median algorithm.

\begin{corollary}
    \label{corollary:median-is-4-competitive}
    The median algorithm is 4-competitive against $\opt$, under the assumption of independent inter-arrival page times.
\end{corollary}

\section{Lower bound for Dominating Distribution Algorithms}
\label{sec:lb}

In this section, we show that there exist problem instances where a dominating distribution algorithm\footnote{Note that we cannot hope to prove a lower bound that applies to \textit{all} dominating distribution algorithms, since the optimal Farthest-in-Future algorithm is technically also a dominating distribution algorithm.} can provably be at least $c$ times worse than $\opt$, for $c \ge 1.5907$. We start with a warm-up analysis that captures the central idea behind the lower bound, and shows a bound of 1.5, before generalizing the analysis to improve the bound. Our lower bound instances throughout are small and simple: they comprise of a Markov chain on $n=3$ pages with a cache size $k=2$.

\subsection{Warm-up: A 1.5-competitive Lower Bound}
\label{sec:1.5-lb}

Consider a simple setting with 3 total pages, and a cache of size 2. At each time step independently, page 1 is requested with probability $1-\eps$, and page 2 and 3 are requested each with probability $\eps/2$. This corresponds to the Markov chain with the transition matrix
\begin{align*}
    \begin{bmatrix}
        1-\eps & \eps/2 & \eps/2 \\
        1-\eps & \eps/2 & \eps/2 \\
        1-\eps & \eps/2 & \eps/2
    \end{bmatrix}.
\end{align*}

We start with the initial cache being $[1,2]$. The reference algorithm $\mcA_{\rf}$ we will compete against always keeps page 1 in its cache. If it suffers a cache miss, it swaps out the other page with the requested page. At any time $t$, the probability that $\mcA_{\rf}$ suffers a miss is therefore exactly $\eps/2$, and hence the expected total number of misses through $T$ timesteps is $\eps T/2$.

Now, consider $\mcA_{\dom}$. We want to capitalize on the fact that $\mcA_\dom$ evicts page 1 with some positive probability, in contrast to $\mcA_\rf$. Namely, observe that 
\begin{align}
    \Pr[\text{$\mcA_\dom$ suffers cache miss at time $t$}] &= \Pr[\text{$1$ in cache of $\mcA_\dom$ at $t$}] \cdot \frac{\eps}{2} + \nonumber\\
    & \Pr[\text{$1$ not in cache of $\mcA_\dom$ at $t$}] \cdot (1-\eps) \nonumber \\
    &= p_t \cdot \frac{\eps}{2} + (1-p_t) \cdot (1-\eps), \label{eqn:misses-A-dom}
\end{align}
where we denote $p_t = \Pr[\text{$1$ in cache of $\mcA_\dom$ at $t$}]$. The latter probability is non-zero for $\mcA_\dom$, which already makes its expected total number of misses larger than that of $\mcA_\rf$. We will therefore aim to maximize this difference.

Observe that 
\begin{align*}
    p_t &= \Pr[\text{$1$ requested at $t-1$}] + \Pr[\text{$1$ in cache of $\mcA_\dom$ at $t-1$}] \cdot \Pr[\text{$1$ not requested at $t-1$}] \cdot \\
    &\Pr[\text{$1$ not evicted by $\mcA_\dom$ at $t-1 ~\big|$ $1$ not requested, $1$ in cache of $\mcA_\dom$ at $t-1$}] \\
    &= 1-\eps + p_{t-1} \cdot \eps \cdot \Pr[\text{$1$ not evicted by $\mcA_\dom$ at $t-1 ~\big|$ $1$ not requested, $1$ in cache of $\mcA_\dom$ at $t-1$}].
\end{align*}
We can analytically calculate the last probability. Say the cache of $\mcA_\dom$ at $t-1$ is, without loss of generality, $[1,2]$. Given that $1$ is not requested, we have that either of $2$ or $3$ are requested, each with probability $1/2$. If $2$ is requested, there is no cache miss. If on the other hand, $3$ is requested, $\mcA_\dom$ calculates a dominating distribution $\mu$ over $[1,2]$. This can be easily calculated. Let $\alpha(p<q)$ again be the probability that $p$ is requested next before $q$ (conditioned on the current page, which in this case, does not matter since the distributions at every time step are independent and identical). We have that
\begin{align*}
    &\alpha(1<1) = \alpha(2<2) = 0, \\
    &\alpha(2<1) = \eps/2 + \eps/2 \cdot \alpha(2<1)
    \implies \quad \alpha(2<1) = \frac{\eps/2}{1-\eps/2}, \\
    &\alpha(1<2) = 1-\eps + \eps/2 \cdot \alpha(1<2)
    \implies \quad \alpha(1<2) = \frac{1-\eps}{1-\eps/2}. 
\end{align*}
The condition for the dominating distribution is that, for every fixed page $q$ in the cache,
\begin{align*}
    \E_{p \sim \mu}[\alpha(p<q)] \le 1/2.
\end{align*}
Letting $\mu(1)=x$ so that $\mu(2)=1-x$, this translates to
\begin{align*}
    &(1-x)\cdot \alpha(2<1) \le 1/2, \quad x\cdot \alpha(1<2) \le 1/2 \\
    \implies \qquad& \frac{3\eps-2}{2\eps} \le x \le \frac{2-\eps}{4-4\eps}.
\end{align*}
We want $x$ to be large as possible, because we want $\mcA_\dom$ to evict $1$ with the highest feasible probability, so that it will suffer a lot of cache misses. Thus, we set $x = \frac{2-\eps}{4-4\eps}$. This gives us that,
\begin{align*}
    &\Pr[\text{$1$ not evicted by $\mcA_\dom$ at $t-1 ~\big|$ $1$ not requested, $1$ in cache of $\mcA_\dom$ at $t-1$}] \\
    &= \frac12 + \frac{1}{2} \cdot (1-x) = 1-\frac{x}{2} = \frac{6-7\eps}{8-8\eps}.
\end{align*}
In total, we get that
\begin{align*}
    p_t &= \underbrace{1-\eps}_{c} + p_{t-1}\cdot \underbrace{\frac{\eps(6-7\eps)}{8-8\eps}}_{d}. 
\end{align*}
Unrolling the recurrence, and using that $p_{1}=1$, we obtain that for $t \ge 2$,
\begin{align*}
    p_t &= c(1+d+\dots+d^{t-2}) + d^{t-1} = c\cdot \frac{1-d^{t-1}}{1-d} + d^{t-1}.
\end{align*}
Summing up $\eqref{eqn:misses-A-dom}$ until $T$, and substituting $p_t$ from above, we get that
\begin{align*}
    \E[\cost(\mcA_\dom, T)] &= T(1-\eps) + \left(\frac{3\eps}{2}-1\right)\sum_{t=1}^{T}p_t \\
    &= T(1-\eps) + \left(\frac{3\eps}{2}-1\right)\left[\frac{c}{1-d}\left(T-\frac{1-d^T}{1-d}\right)+\frac{1-d^T}{1-d}\right].
\end{align*}
Taking the limit as $T \to \infty$, we get that
\begin{align*}
    \lim_{T \to \infty}\frac{\E[\cost(\mcA_\dom, T)]}{\E[\cost(\mcA_\rf, T)]} &= \frac{1-\eps+\left(\frac{3\eps}{2}-1\right)\left(\frac{c}{1-d}\right)}{\eps/2}.
\end{align*}
Substituting the values of $c$ and $d$, and taking the limit as $\epsilon \to 0$, the ratio above converges to $1.5$.

\subsection{A Slightly More General Analysis that Gives a Better Lower Bound}
\label{sec:1.5907-lb}

In the instance above, instead of assigning $\eps/2$ mass equally to pages 2 and 3, we can also imagine splitting the $\eps$ mass unequally between these pages. Namely, let the underlying distribution on the pages be such that, at each time step, page 1 is requested with probability $1-\eps$, page 2 is requested with probability $\eps_1$ and page 3 is requested with probability $\eps-\eps_1$, where the regime we want to keep in mind is $1-\eps \ge \eps_1 \ge \eps-\eps_1$. Note that $\eps_1=\eps/2$ corresponds to the instance above. Let the initial cache be $[1,2]$.

The reference algorithm $\mcA_\rf$ still always maintains page 1 in its cache. Let
\begin{align*}
    p_t &= \Pr[\text{$\mcA_\rf$'s state is $[1,2]$ at $t$}].
\end{align*}
We have that
\begin{align*}
    p_t = p_{t-1} \cdot (1-\eps+\eps_1) + (1-p_{t-1}) \cdot \eps_1 
    = \eps_1 + p_{t-1} \cdot (1-\eps) 
    = \eps_1 \left(\frac{1-(1-\eps)^{t-1}}{\eps}\right) + (1-\eps)^{t-1},
\end{align*}
where we unrolled the recurrence in the last step, and used that $p_1=1$. Then,
\begin{align*}
    \Pr[\mcA_\rf \text{ suffers miss at }t] = p_t \cdot (\eps-\eps_1) + (1-p_t) \cdot \eps_1
    = \eps_1 + (\eps-2\eps_1)\cdot p_t,
\end{align*}
and hence
\begin{align}
    \E[\cost(\mcA_\rf, T)] &= T\eps_1 + (\eps-2\eps_1) \sum_{t=1}^T p_t \nonumber \\
    &= T\eps_1 + (\eps-2\eps_1)\left[ \frac{\eps_1}{\eps}\left( T - \frac{1-(1-\eps)^T}{\eps} \right) + \frac{1-(1-\eps)^T}{\eps} \right]. \label{eqn:ref-gen-misses}
\end{align}
Now, we analyze the behavior of $\mcA_\dom$. Let
\begin{align*}
    &p_t = \Pr[\text{$\mcA_\dom$'s state is $[1,2]$ at $t$}] \\
    &q_t = \Pr[\text{$\mcA_\dom$'s state is $[1,3]$ at $t$}] \\
    &r_t = \Pr[\text{$\mcA_\dom$'s state is $[2,3]$ at $t$}],
\end{align*}
where $p_1=1, q_1=0, r_1=0$. Then, we have that
\begin{align}
    \label{eqn:miss-eqn}
    \Pr[\mcA_\dom \text{ suffers miss at }t] &= p_t \cdot (\eps-\eps_1) + q_t \cdot \eps_1 + r_t \cdot (1-\eps).
\end{align}
 Observe further that
\begin{align*}
    p_t &= p_{t-1} \cdot \Pr[\text{1 or 2 requested at }t-1] + q_{t-1} \cdot \eps_1 \cdot \Pr[\text{$\mcA_\dom$ evicts 3 from }[1,3]] \\
    &\qquad+ r_{t-1} \cdot (1-\eps) \cdot \Pr[\text{$\mcA_\dom$ evicts 3 from }[2,3]] \\
    &= p_{t-1} \cdot (1-\eps+\eps_1) +  q_{t-1} \cdot \eps_1 \cdot \Pr[\text{$\mcA_\dom$ evicts 3 from }[1,3]] \\
    &\qquad+ r_{t-1} \cdot (1-\eps) \cdot \Pr[\text{$\mcA_\dom$ evicts 3 from }[2,3]] \\
    q_t &= q_{t-1} \cdot \Pr[\text{1 or 3 requested at }t-1] + p_{t-1} \cdot (\eps-\eps_1) \cdot \Pr[\text{$\mcA_\dom$ evicts 2 from }[1,2]] \\
    &\qquad+ r_{t-1} \cdot (1-\eps) \cdot \Pr[\text{$\mcA_\dom$ evicts 2 from }[2,3]] \\
    &= q_{t-1} \cdot (1-\eps_1) + p_{t-1} \cdot (\eps-\eps_1) \cdot \Pr[\text{$\mcA_\dom$ evicts 2 from }[1,2]] \\
    &\qquad+ r_{t-1} \cdot (1-\eps) \cdot \Pr[\text{$\mcA_\dom$ evicts 2 from }[2,3]] \\
    r_t &= r_{t-1} \cdot \Pr[\text{2 or 3 requested at }t-1] + p_{t-1} \cdot (\eps-\eps_1) \cdot \Pr[\text{$\mcA_\dom$ evicts 1 from }[1,2]] \\
    &\qquad+ q_{t-1} \cdot \eps_1 \cdot \Pr[\text{$\mcA_\dom$ evicts 1 from }[1,3]] \\
    &= r_{t-1} \cdot \eps + p_{t-1} \cdot (\eps-\eps_1) \cdot \Pr[\text{$\mcA_\dom$ evicts 1 from }[1,2]]  + q_{t-1} \cdot \eps_1 \cdot \Pr[\text{$\mcA_\dom$ evicts 1 from }[1,3]].
\end{align*}
We can calculate each of the eviction probabilities above while respecting the properties that $\mcA_\dom$ must satisfy. Namely, consider that the cache state is $[1,2]$. Let $\mu$ be the dominating distribution that $\mcA_\dom$ constructs, where $\mu(1)=x$ and $\mu(2)=1-x$. We have that
\begin{align*}
    \alpha(2<1) &= \eps_1 + (\eps-\eps_1)\alpha(2<1) \implies\quad \alpha(2<1) = \frac{\eps_1}{1-\eps+\eps_1} \\
    \alpha(1<2) &= 1-\alpha(2<1) = \frac{1-\eps}{1-\eps+\eps_1}.
\end{align*}
The dominating distribution condition requires that
\begin{align*}
    &(1-x)\cdot\alpha(2<1) \le \frac{1}{2}, \quad x\cdot\alpha(1<2) \le \frac{1}{2} \\ 
    \implies \qquad & \frac{\alpha(2<1)-1/2}{\alpha(2<1)} \le x \le \frac{1}{2\alpha(1<2)}.
\end{align*}
We choose to set the upper condition to be tight, which results in
\begin{align*}
    x &= \Pr[\text{$\mcA_\dom$ evicts 1 from }[1,2]] = \frac{1-\eps+\eps_1}{2-2\eps}, \\
    1-x &= \Pr[\text{$\mcA_\dom$ evicts 2 from }[1,2]] = \frac{1-\eps-\eps_1}{2-2\eps}.
\end{align*}
Now, consider the same calculation, where the cache state is $[1,3]$, $\mu(1)=x$ and $\mu(3)=1-x$. We have that
\begin{align*}
    \alpha(3<1) &= \eps-\eps_1 + \eps_1 \cdot \alpha(3<1) \implies\quad \alpha(3<1) = \frac{\eps-\eps_1}{1-\eps_1} \\
    \alpha(1<3) &= 1-\alpha(3<1) = \frac{1-\eps}{1-\eps_1}.
\end{align*}
The dominating distribution condition requires
\begin{align*} 
    \frac{\alpha(3<1)-1/2}{\alpha(3<1)} \le x \le \frac{1}{2\alpha(1<3)}.
\end{align*}
Again, we make the upper bound tight to get
\begin{align*}
    x &= \Pr[\text{$\mcA_\dom$ evicts 1 from }[1,3]] = \frac{1-\eps_1}{2-2\eps}, \\
    1-x &= \Pr[\text{$\mcA_\dom$ evicts 3 from }[1,3]] = \frac{1+\eps_1-2\eps}{2-2\eps}.
\end{align*}
Finally, consider the calculation, where the cache state is $[2,3]$, $\mu(2)=x$ and $\mu(3)=1-x$. We have that
\begin{align*}
    \alpha(3<2) &= \eps-\eps_1 + (1-\eps) \cdot \alpha(3<2) \implies\quad \alpha(3<2) = \frac{\eps-\eps_1}{\eps} \\
    \alpha(2<3) &= 1-\alpha(3<2) = \frac{\eps_1}{\eps}.
\end{align*}
The dominating distribution condition requires
\begin{align*} 
    \frac{\alpha(3<2)-1/2}{\alpha(3<2)} \le x \le \frac{1}{2\alpha(2<3)},
\end{align*}
so we set
\begin{align*}
    x &= \Pr[\text{$\mcA_\dom$ evicts 2 from }[2,3]] = \frac{\eps}{2\eps_1}, \\
    1-x &= \Pr[\text{$\mcA_\dom$ evicts 3 from }[2,3]] = \frac{2\eps_1-\eps}{2\eps_1}.
\end{align*}
Substituting all these eviction probabilities, and writing the system of equations for $p_t,q_t,r_t$ above in matrix form, we obtain
\begin{align*}
    \begin{bmatrix}
        p_t \\ q_t \\ r_t
    \end{bmatrix}
    &=
    \begin{bmatrix}
       1-\eps+\eps_1  & \eps_1 \cdot \frac{1+\eps_1-2\eps}{2-2\eps} & (1-\eps) \cdot \frac{2\eps_1-\eps}{2\eps_1} \\
       (\eps-\eps_1) \cdot \frac{1-\eps-\eps_1}{2-2\eps} & 1-\eps_1 & (1-\eps) \cdot \frac{\eps}{2\eps_1} \\
       (\eps-\eps_1) \cdot \frac{1-\eps+\eps_1}{2-2\eps} & \eps_1 \cdot \frac{1-\eps_1}{2-2\eps} & \eps
    \end{bmatrix}
    \begin{bmatrix}
        p_{t-1} \\ q_{t-1} \\ r_{t-1}
    \end{bmatrix} \\
    &= B^{t-1} \begin{bmatrix}
        p_{1} \\ q_{1} \\ r_{1}
    \end{bmatrix}
    = B^{t-1} \begin{bmatrix}
        1 \\ 0 \\ 0
    \end{bmatrix},
\end{align*}
where we denoted the matrix by $B$ and rolled out the recurrence. Substituting in \eqref{eqn:miss-eqn}, we have that
\begin{align*}
    \Pr[\mcA_\dom \text{ suffers miss at }t] &= 
    \begin{bmatrix}
        \eps-\eps_1 & \eps_1 & 1-\eps
    \end{bmatrix}
    B^{t-1} \begin{bmatrix}
        1 \\ 0 \\ 0
    \end{bmatrix},
\end{align*}
and thus,
\begin{align}
    \E[\cost(\mcA_\dom, T)] &= \begin{bmatrix}
        \eps-\eps_1 & \eps_1 & 1-\eps
    \end{bmatrix}
    (I + B + B^2 + \dots + B^{T-1}) \begin{bmatrix}
        1 \\ 0 \\ 0
    \end{bmatrix}. \label{eqn:dom-gen-misses}
\end{align}
Now that we have exact expressions given by \eqref{eqn:ref-gen-misses} and \eqref{eqn:dom-gen-misses}, we can do a computer search over assignments to $T, \eps, \eps_1$---each assignment gives a lower bound for $\frac{\E[\cost(\mcA_\dom, T)]}{\E[\cost(\mcA_\rf, T)]}$. For example, our (admittedly limited) search yields
\begin{align*}
    \frac{\E[\cost(\mcA_\dom, T)]}{\E[\cost(\mcA_\rf, T)]} \ge 1.5907,
\end{align*}
for $\eps=10^{-5}, T=10^8, \eps_1=0.7069\eps$.

\begin{remark}
    \label{remark:potential-improving-lower-bound}
    While we do believe that further tuning $\eps, \eps_1$ and $T$ will likely not improve the lower bound by too much, it is conceivable that a better lower bound may be obtained by going beyond $n=3, k=2$ (albeit with messier calculations). For example, there appears to be scope for improvement within instances having the same blueprint as above for slightly larger values of $n$ and $k$, where there are certain high-probability pages that $\mcA_\rf$ always keeps in the cache.
\end{remark}

\section{Learning the Markov Chain from Samples}
\label{sec:robustness}

It is worth pointing out that the form of \Cref{lem:tight-analysis} allows us to obtain performance guarantees even when we only have approximate estimates of the $\alpha(p<q)$ values, as also noted in passing by \cite{lund1999paging}. In particular, suppose that we only have $\eps$-approximate (multiplicate or additive) estimates $\hat{\alpha}(p<q)$ of $\alpha(p<q)$, that still satisfy $\hat{\alpha}(p<q) + \hat{\alpha}(q<p)=1$, and $\hat{\alpha}(p<p)=0$. Suppose that at each cache miss, we use these $\hat{\alpha}(p<q)$ values to compute a dominating distribution $\hat{\mu}$, and draw the page $p$ to evict from $\hat{\mu}$.

Even with such an approximate dominating distribution, we can still guarantee either $\frac{2}{1-2\eps}$-competitiveness (in the case that $\hat{\alpha}(p<q) \in [\alpha(p<q)-\eps, \alpha(p<q)+\eps]$), or $\frac{2-2\eps}{1-2\eps}$-competitiveness (in the case that $\hat{\alpha}(p<q) \in [(1-\eps)\alpha(p<q), (1+\eps)\alpha(p<q)]$). This follows as a direct consequence of \Cref{claim:dominating-distribution-probability} and \Cref{lem:tight-analysis}.

Perhaps the first scenario to consider here is when we do not assume prior knowledge of the transition matrix $M$, but only have access to samples previously drawn from the Markov chain. For example, we can imagine that the algorithm designer has available a large training dataset of past page requests. This dataset can be used to derive estimates of the entries in the transition matrix $M$, for instance, using one of the several existing estimators \citep{hao2018learning, wolfer2021statistical, huang2024non}. In fact, recall that the computation of the $\alpha(p<q)$ values requires only for us to solve a linear system determined by the entries in the transition matrix (\Cref{eqn:linear-system-for-alpha}). Hence, if we have accurate estimates for the entries in $M$, we can use the theory of perturbed linear systems to bound the resulting error in the estimates of the $\alpha(p<q)$ values. Thereafter, using the argument from the previous paragraph, we can obtain a performance guarantee for the competitiveness of a dominating distribution algorithm that uses these estimates.

We can turn the above informal argument into a formal learning-theoretic result. Formally, suppose that every entry in $M$ is at least some $\delta > 0$. In this case, the Markov chain is irreducible (meaning that there is a positive, and unique probability of eventually getting from any state to another). Thus, in this case, for every fixed $p,q$, the linear system \eqref{eqn:linear-system-for-alpha} has a (unique) solution, meaning that matrix associated with it is invertible. The relevant quantity here will be the conditioning of the \textit{worst} linear system that we might ever have to solve. Towards this, let $L_{p,q}$ be the matrix associated with the linear system in \eqref{eqn:linear-system-for-alpha} for pages $p,q$, and define:
\begin{align}
    \label{eqn:gamma-def}
    \gamma &= \sup_{p \neq q} \|L_{p,q}^{-1}\|_\infty,
\end{align}
where for a square matrix $A \in \R^{n \times n}$, $\|A\|_\infty = \sup_i \sum_{j}|A_{i,j}|$. Observe from \eqref{eqn:linear-system-for-alpha} that $1 \le \|L_{p,q}\|_\infty \le 2$ for every $p,q$.

We choose to adopt the following result of \cite{hao2018learning} for our exposition:

\begin{theorem}[Theorem 4 in \cite{hao2018learning}]
    \label{thm:learning-markov-chains}
    Suppose we unroll a Markov chain to obtain $m$ samples $X_1,\dots,X_m$, where the initial distribution of $X_1$ is arbitrary, but the transition matrix $M \in [0,1]^{n \times n}$ of the Markov chain satisfies $M_{i,j} \ge \delta > 0$ for every $i,j$. Then, there exists an estimator $\hat{M}=\hat{M}(X_1,\dots,X_m)$ that satisfies
    \begin{align*}
        \E_{X_1,\dots,X_m}\|M(i,:) - \hat{M}(i,:)\|_2^2 \le O \left(\frac{1}{m\delta}\right)
    \end{align*}
    for every $i \in [n]$, where $M(i,:)$ (respectively $\hat{M}(i,:)$) denotes the $i^\text{th}$ row of $M$ (respectively $\hat{M}$).
\end{theorem}

Using this theorem, %
we can obtain the following result:

\begin{theorem}
    \label{thm:approx-lpr-from-samples}
    Suppose that the page requests are generated from an unknown Markov chain where every entry in the transition matrix $M$ is at least $\delta > 0$. Let $\gamma$ be as defined in \eqref{eqn:gamma-def}. Given a training dataset of $m = O\left(\frac{\gamma^2n^2}{\eps^2\delta}\right)$ past page requests from $M$, there is an algorithm $\mcA$ which is $\frac{2}{1-2\eps}$-competitive against $\opt$ with probability at least $0.99$ over the $m$ samples.. %
\end{theorem}
\begin{proof}
    From the initial $m$ page requests, $\mcA$ constructs the estimator $\hat{M}$ given by \Cref{thm:learning-markov-chains}. Applying Markov's inequality to the in-expectation guarantee on $\hat{M}$, for every fixed $i \in [n]$, we have that with probability at least $1-0.01 \cdot \frac{1}{n}$, 
    \begin{align*}
        \|M(i,:) - \hat{M}(i,:)\|_2^2 \le O \left(\frac{n}{m\delta}\right).
    \end{align*}
    A union bound over the universe of $n$ pages gives that with probability at least $0.99$, we have that 
    \begin{align*}
        \sup_{i \in [n]}\|M(i,:) - \hat{M}(i,:)\|_2^2 \le O \left(\frac{n}{m\delta}\right)
    \end{align*}
    The Cauchy-Schwarz inequality then implies that with probability at least $0.99$,
    \begin{align*}
        \sup_{i \in [n]}\sum_{j=1}^n |M(i,j)-\hat{M}(i,j)| \le O \left(\frac{n}{\sqrt{m\delta}}\right).
    \end{align*}
    Denoting $\Delta = \hat{M}-M$, we have effectively argued that $\|\Delta\|_\infty \le O \left(\frac{n}{\sqrt{m\delta}}\right)$ with probability at least $0.99$.

    The algorithm $\mcA$ plays the dominating distribution algorithm as its eviction policy on a cache miss; however, it constructs the dominating distribution using the estimator $\hat{M}$. Namely, for every pair $p,q$, let $\hat{L}_{p,q}$ be the matrix associated with the linear system in \eqref{eqn:linear-system-for-alpha}, where instead of the entries in $M$, we plug in entries from our estimator $\hat{M}$. Consider the solution $\hat{x}$ of the perturbed linear system $\hat{L}_{p,q} \cdot \hat{x} = b_{p,q}$, where $b_{p,q}$ is the right-hand side of \eqref{eqn:linear-system-for-alpha}. We care about bounding the difference between $\hat{x}$ and $x$, where $x$ is the solution to $L_{p,q} \cdot x = b_{p,q}$. Note that the right-hand side in both the perturbed and unperturbed systems (namely $b_{p,q}$) is still the same.

    From a standard analysis of solutions to perturbed linear systems (e.g., see Theorem 1 in \citep{falknotes}), we have that
    \begin{align*}
        \|\hat{x}-x\|_\infty &\le \frac{\|L_{p,q}^{-1}\|_\infty \|\hat{L}_{p,q}-L_{p,q}\|_\infty \|x\|_\infty}{1-\|L_{p,q}^{-1}\|_\infty \|\hat{L}_{p,q}-L_{p,q}\|_\infty}.
    \end{align*}
    Now, observe that $\|\hat{L}_{p,q}-L_{p,q}\|_\infty \le \|\Delta\|_\infty$, and also that $\|x\|_\infty=1$. Plugging in our bound above on $\Delta$, we get
    \begin{align*}
        \|\hat{x}-x\|_\infty &\le \frac{O(\gamma n/\sqrt{m\delta})}{1-O(\gamma n/\sqrt{m\delta})}.
    \end{align*}
    Setting $m=C \cdot \frac{\gamma^2n^2}{\eps^2\delta}$ for a constant $C$ is sufficient to make the right-hand side above to be at most $\eps$.\footnote{In fact, this is also sufficient to satisfy the technical condition of Theorem 1 in \citep{falknotes} that we used.} But this means that all the $\hat{\alpha}(p<q)$ estimates that $\mcA$ constructs satisfy
    \begin{align*}
        {\alpha}(p<q)-\eps \le \hat{\alpha}(p<q) \le {\alpha}(p<q) + \eps.
    \end{align*}
    Thus, the dominating distribution $\hat{\mu}$ that $\mcA$ constructs at every cache miss has the property that for every fixed $q$ in the cache,
    \begin{align*}
        \E_{p \sim \hat{\mu}}[\hat{\alpha}(p<q)] \le \frac{1}{2} 
        \quad\implies\quad \E_{p \sim \hat{\mu}}[{\alpha}(p<q)] \le \frac{1}{2}+\eps.
    \end{align*}
    Instantiating \Cref{claim:dominating-distribution-probability} and \Cref{lem:tight-analysis} then gives the desired result.
\end{proof}

\section*{Acknowledgements}
Chirag is supported by Moses Charikar and Gregory Valiant's Simons Investigator Awards. The authors would like to thank Avrim Blum for helpful pointers. The authors also thank Romain Cosson for pointing out that the guarantee for the median algorithm in \cite{lund1994ip} requires an additional independence assumption.

\bibliographystyle{abbrvnat}
\bibliography{ref}

\end{document}